\documentclass[graybox, 11pt]{svmult}

\usepackage{mathptmx}       
\usepackage{helvet}         
\usepackage{courier}        
\usepackage{makeidx}         
\usepackage{graphicx}        
\usepackage{multicol}        

\usepackage{amsmath}
\usepackage{amssymb}
\usepackage{amsfonts}

\newcommand{\ep}{\varepsilon}

\newcommand{\R}{\mathbb{R}}

\newcommand{\PP}{\mathbb{P}}
\newcommand{\EE}{\mathbb{E}}

\newcommand{\LL}{\mathcal{L}}

\newcommand{\mbf}{\mathbf}

\newtheorem{assumption}{Assumption}

\renewcommand{\baselinestretch}{1.2}
\begin{document}

\title*{Fick's Law for the Lorentz Model in a weak coupling regime}
\author{Alessia Nota}
\institute{Alessia Nota \at Dipartimento di Matematica, Universit\`a di Roma La Sapienza\\ 
Piazzale Aldo Moro 5, 00185 Roma -- Italy ,
 \email{nota@mat.uniroma1.it}}
%
%
\maketitle

\abstract{In this paper we deal with further recent developments, strictly connected to 
the result obtained in \cite{BNPP}. 
We consider the Lorentz gas out of equilibrium in a weak coupling regime. 
Each obstacle of the Lorentz gas generates a smooth radially symmetric potential with compact support.  We prove that the macroscopic current in the stationary state is given by the Fick's law of diffusion. 
The diffusion coefficient is given by the Green-Kubo formula associated to the generator of the diffusion process dictated by the linear Landau equation. }

\section{Introduction}

The understanding of transport phenomena of nonequilibrium thermodynamics starting form the microscopic dynamics 
is one of the most challenging problem in statistical mechanics. 

Nonequilibrium stationary states describe the state of a mechanical system driven and maintained out of equilibrium. 
Their main characteristic 
is that they commonly exhibit transport phenomena. They sustain steady flows (e.g. energy flow, particles flow or momentum flow) and the usually conserved quantities ( mass, momentum and energy) flow in response to a gradient. For instance the heat flow and the mass flow appear in response to a temperature gradient and a concentration gradient respectively. These processes are well described by phenomenological linear laws, the Fourier's and Fick's law respectively.

In the current literature there are very few rigorous results concerning the derivation of the these phenomenological laws from a microscopic model (see for instance [LS], [LS1], [LS2]). A contribution in this direction is the validation of the Fick's law for the Lorentz model in a low density situation which has been recently proven in \cite{BNPP}.
To consider the system out of equilibrium, in \cite{BNPP}, they study the Lorentz gas in a bounded region in the plane and couple the system with two mass reservoirs at the boundaries. 
More precisely they consider the slice $\Lambda= (0,L)\times \R$ in the plane. In the left half plane there is a free gas of light particles at density $\rho_1$, in the right half plane there is a free gas of light particles at density $\rho_2$ which play the role of mass reservoirs. The light particles are not interacting among themselves. Inside $\Lambda$ there is a Poisson distribution of intensity $\mu$ of hard core scatterers. The light particles flow through the boundaries and are 
elastically reflected by the scatterers.  
For this model they prove the existence of a stationary state for which
\begin{equation}
\label{BNPPL:FL}
J \approx -D \nabla \rho 
\end{equation}
where $J$ is the mass current, $\rho$ is the mass density and $D>0$ is the diffusion coefficient. Formula \eqref{BNPPL:FL} is the well known Fick's law whose validity has been proven in \cite{BNPP}.
We remind that according to the low-density regime considered they can use the linear Boltzmann equation as a bridge between the original mechanical system and the diffusion equation. This strategy works since they provide an explicit control of the error in the kinetic limit which suggests the scale of times for which the diffusive limit can be achieved. 
The result is presented in a two dimensional setting but it holds in dimension higher than two. The two dimensional case is the most interesting to analyze since the pathologic configurations preventing the Markovianity on a kinetic scale are harder to estimate in this case.  

We may wonder if the same result could be achieved if we slightly modify the model. We consider the same geometry described above but inside $\Lambda$ now we have a Poisson distribution of scatterers which are no longer hard cores. We assume that each obstacle generates a smooth, radial, short-range potential. In the same spirit as in \cite{BNP}, \cite{ESY}, we scale 
%
the range of the interaction and the density of the scatterers according to
\begin{equation}\label{BNPPL:INTRO:WC}
\left.\begin{array}{ll}
\phi_\ep(x)={\ep}^\alpha\phi(\frac{x}{\ep}),&\\\vspace{1.5mm}
\mu_{\ep}=\ep^{-(2\alpha+\lambda+1)}\mu,&
\end{array}\right.
\end{equation}
with $\alpha\in (0,\frac 1 2)$ and $\lambda>0$. 

The scaling \eqref{BNPPL:INTRO:WC} means that the kinetic regime describes the system for kinetic times $O(1)$ (i.e. $\lambda=0$). Observe that when $\lambda=0$ the limiting cases $\alpha=0$ and $\alpha=1/2$ correspond respectively to the low density limit and the weak-coupling limit. In the intermediate scale between the low density and the weak-coupling regime the kinetic equation that appears in the limit is the linear Landau equation. 
One can go further to diffusive times provided that $\lambda>0$ is not too large. The intermediate level of description between the mechanical system and the diffusion equation is given by the linear Landau equation with a divergent factor in front of the collision operator.
Since the scale of time for which the system diffuses should not prevent the Markov property, there is a constraint on $\lambda$. More precisely there exists a threshold $\lambda_0=\lambda(\alpha)$, emerging from the explicit estimate of the set of pathological configurations producing memory effects,
s.t. for $\lambda<\lambda(\alpha)$, the microscopic solution of the time dependent problem converges to the solution of the heat equation in the limit $\ep\to 0$.  
We refer to \cite{BNP}, Section 6, for further details. The result mentioned above concerns the time dependent problem. 
In this paper we deal with the stationary problem and provide a rigorous derivation of Fick's law of diffusion for this model. We prove that there exists a unique stationary solution for the microscopic dynamics which converges to the stationary solution of the heat equation, namely to the linear profile of the density.
We underline that in order to obtain the stationary solution of the microscopic dynamics we need to characterize the stationary solution of the linear Landau equation. To handle this problem 
we will use the analysis of the time dependent problem and the explicit solution of the heat equation.

\section{The model and main results}\label{BNPPL:sec2}
Let $\Lambda \subset \R^2$ be the strip $(0,L)\times \R$.
We consider a Poisson distribution of fixed disks (scatterers) of radius $\ep$ in $\Lambda$ and denote by $c_1,\dots,c_N\in\Lambda$ their centers. 
This means that, given $\mu>0$, the probability density of finding $N$ obstacles in a bounded measurable set $A\subset\Lambda$ is 
\begin{equation}\label{BNPPL:poisson}
\PP(\,d\mbf{c}_{N})=e^{-\mu |A|}\frac{\mu^N}{N!}\,dc_1\dots\,dc_N
\end{equation}
where $|A|=\text{meas}A$ and $\mbf{c}_{N}=(c_1,\dots, c_N)$. Since the modulus of the velocity of the test particle is constant, we assume it to be equal to one, so that the phase space of our system is $\Lambda\times S_1$.

We rescale the intensity $\mu$ of the obstacles as 
$$
\mu_{\varepsilon}=\varepsilon^{-2\alpha-1}\ep^{-\lambda}\mu,\quad \alpha\in (0,1/8),\,\lambda>0
$$
where, from now on, $\mu>0$ is fixed. More precisely we make the following assumption.
\begin{assumption}\label{BNPPL:A1}
We set $\gamma=1-8(\alpha+\lambda/2)$, the parameter $\lambda$ is such that as $\varepsilon \to 0$,
\begin{equation}\label{BNPPL:assump}
\ep^{\gamma-4\lambda}\to 0,
\end{equation}
namely $\lambda<\frac{1-8\alpha}{8}$.
\end{assumption}

Accordingly, we denote by $\PP_{\ep}$  the probability density \eqref{BNPPL:poisson} with $\mu$ replaced by $\mu_\ep$. $\EE_{\ep}$ will be the expectation with respect to the measure $\PP_{\ep}$.

We now introduce a radial potential $\phi(r)$ such that
\begin{itemize}
\item $\phi\in C^2([0,1])$,
\item $\phi(0)>0$ and $r\to\phi( r)$ is strictly decreasing in $[0,1]$.
\end{itemize}
We rescale the intensity of the interaction potential as
$$\phi\rightarrow\ep^{\alpha}\phi.$$
Then the Equations of motion are 
\begin{equation}\label{BNPPL:eqscaled}
\left\{\begin{array}{ll}
\dot{x}=v&\\
\dot{v}=-\ep^{\alpha-1}\sum_{i}\nabla\phi(\frac{|x-c_i|}{\ep})&.
\end{array}\right.
\end{equation}
For a given configuration of obstacles $\mbf{c}_N$, we denote by $T^{-t}_{\mbf{c}_{N}}(x,v)$ the (backward) 
flow, solution of \eqref{BNPPL:eqscaled}, with initial datum $(x,v)\in\Lambda\times S_1$ and define
$t-\tau$, $\tau=\tau(x,v,t,\mbf{c}_N)$, as the first (backward) hitting time with the boundary. We use the notation $\tau=0$ to indicate the event such that the trajectory $T^{-s}_{\mbf{c}_{N}}(x,v)$, $s\in [0,t]$, never hits the boundary.
For any $t\geq 0$ the one-particle correlation function reads
\begin{equation}\label{BNPPL:def:fep}
f_{\ep}(x,v,t)=\EE_\varepsilon[f_B (T^{-(t-\tau)}_{\mbf{c}_{N}}(x,v))\chi(\tau>0)] + \EE_\ep[f_0 (T^{-t}_{\mbf{c}_{N}}(x,v))\chi(\tau=0)],
\end{equation} 
where $f_0\in L^\infty(\Lambda\times S_1)$ and the boundary value $f_B$ is defined by
\begin{equation*}
f_B(x,v):=\left\{\begin{array}{ll}
\rho_1 M(v)\quad\text{if}\quad x\in \{0\}\times\R,\quad v_1>0,&\vspace{3mm} \\
\rho_2 M(v)\quad\text{if}\quad x\in \{L\}\times\R,\quad v_1<0,& \vspace{3mm}
\end{array}\right.
\end{equation*}
with $M(v)$ the density of the uniform distribution on $S_1$ and $\rho_1, \rho_2>0$.  Here $v_1$ denotes the horizontal component of the velocity $v$.
Without loss of generality we assume $\rho_2>\rho_1$.
Since $M(v)=\frac{1}{2\pi}$, from now on we will absorb it in the definition of the boundary values $\rho_1, \rho_2$. Therefore we set
\begin{equation}\label{BNPPL:def:fB}
f_B(x,v):=\left\{\begin{array}{ll}
\rho_1\quad\text{if}\quad x\in \{0\}\times\R,\quad v_1>0,&\vspace{3mm} \\
\rho_2\quad\text{if}\quad x\in \{L\}\times\R,\quad v_1<0.& \vspace{3mm}
\end{array}\right.
\end{equation}



We are interested in the stationary solutions $f_{\ep}^S$ of the above problem. More precisely $f_\ep^S(x,v)$ solves  
\begin{equation}\label{BNPPL:def:ST}
f_{\ep}^S(x,v)=\EE_\varepsilon[f_B (T^{-(t-\tau)}_{\mbf{c}_{N}}(x,v))\chi(\tau>0)] + \EE_\ep[f_ \ep^S(T^{-t}_{\mbf{c}_{N}}(x,v))\chi(\tau=0)]. 
\end{equation} 

The main result of the present paper can be summarized in the following theorem.

\begin{theorem}\label{BNPPL:th:MAIN1}
For $\ep$ sufficiently small there exists a unique $L^\infty$ stationary solution $f_\ep^S$ for the microscopic dynamics (i.e. satisfying \eqref{BNPPL:def:ST}). Moreover, as $\varepsilon \to 0$ 
\begin{equation}\label{BNPPL:convergenzaSTAZ}
f_\ep^S\rightarrow \varrho^S,
\end{equation}
where $\varrho^S$ is the stationary solution of the heat equation with the following boundary conditions
\begin{equation}\label{BNPPL:HEAT}
\left\{
 \begin{array}{l}\vspace{0.2cm}
\varrho^S(x)=\rho_1,\ \ \ \ \ x\in \{0\}\times\R,  
\vspace{0.2cm}\\
\varrho^S(x)=\rho_2,\ \ \ \ \ x\in \{L\}\times\R.
\end{array} \right.
\end{equation}
The convergence is in $L^{2}((0,L)\times S_1)$.
\end{theorem}

Some remarks on the above Theorem are in order. 
The boundary conditions of the problem
depend on the space variable only through the horizontal component. As a consequence, the stationary solution $f_\ep^S$ of the microscopic problem, as well as the stationary solution $\varrho^S$ of the heat equation, inherits the same feature. This justifies the convergence in $L^{2}((0,L)\times S_1)$ instead of in $L^{2}(\Lambda\times S_1)$. 
The explicit expression for the stationary solution $\varrho^S$ reads 
\begin{equation}\label{BNPPL:statheat}
\varrho^S(x)= \frac{\rho_1(L-x_1)+\rho_2 x_1}{ L},
\end{equation}
where $x_1$ is the horizontal component of the space variable $x$. We note that in order to prove Theorem \ref{BNPPL:th:MAIN1} it is enough to assume that $\ep^{\gamma-3\lambda}\to 0,$ i.e. $\lambda<\frac{1-8\alpha}{7}.$ 
The stronger Assumption \ref{BNPPL:A1} is needed to prove Theorem \ref{BNPPL:th:MAIN2} below.\\\vspace{2mm}

Next, to discuss the Fick's law, we introduce the stationary mass flux 
\begin{equation}\label{BNPPL:def:j}
J_\ep^S(x)=\ep^{-\lambda}\int_{S_1}v\,f_\ep^S(x,v)\,dv,
\end{equation}
and the stationary mass density
\begin{equation}\label{BNPPL:def:mass}
\varrho_\ep^S(x)=\int_{S_1}f_\ep^S(x,v)\,dv.
\end{equation}
Note that $J_\ep^S$ is the total amount of mass flowing through a unit area in a unit time interval. Although in a stationary problem there is no typical time scale, the factor $\ep^{-\lambda}$ appearing in the definition of $J_\ep^S$, is reminiscent of the time scaling necessary to obtain a diffusive limit.

\begin{theorem}[\textbf{Fick's law}] \label{BNPPL:th:MAIN2}
We have 
\begin{equation}\label{BNPPL:FickL}
J_\ep^S+D\nabla_x\varrho_\ep^S\to 0
\end{equation}
as $\ep\to 0$. The convergence is in $\D'(0,L)$ and $D>0$ is given by the Green-Kubo formula 
\begin{equation}\label{BNPPL:GK} 
D=\frac{2}{\mu}\int_{S_1}{v\cdot\big(-\Delta_{|v|}^{-1}\big)v\,dv}.
\end{equation}
Moreover
\begin{equation}\label{BNPPL:Jlim}
J^S=\lim_{\ep\to 0}J_\ep^S(x),
\end{equation}
where the convergence is in $L^2(0,L)$
and 
\begin{equation}\label{BNPPL:FICK}
J^S=-D\,\nabla\varrho^S=-D\,\frac{\rho_2-\rho_1}{L},
\end{equation}
where $\varrho^S$ is the linear profile \eqref{BNPPL:statheat}.
\end{theorem}
\vspace{4mm}
Observe that, as expected by physical arguments, the stationary flux $J^S$ does not depend on the space variable. Furthermore the diffusion coefficient $D$ is determined by the behavior of the system at equilibrium and in particular it is equal to the diffusion coefficient for the time dependent problem.

%
%

\section{Proofs}\label{BNPPL:proofs}
In order to prove Theorem \ref{BNPPL:th:MAIN1} our strategy is the following. 
We introduce the stationary linear Landau equation 
\begin{equation}\label{BNPPL:eq:LandauS}
\left\{\begin{array}{ll}
\big(v\cdot\nabla_x\big)g_\ep^S(x,v)=\ep^{-\lambda}\,\LL  g_\ep^S(x,v),
&\vspace{0.2cm}\\
g_\ep^S(x,v)=\rho_1,\ \ \ \ \ x\in \{0\}\times \R,\quad v_1>0, 
&\vspace{0.2cm}\\
g_\ep^S(x,v)=\rho_2,\ \ \ \ \ x\in \{L\}\times \R,\quad v_1<0, \
&
\end{array}\right.
\end{equation}
where $\LL=\frac{\mu}{2}\Delta_{|v|}$ and $\Delta_{|v|}$ is the Laplace Beltrami operator on the circle of radius ${|v|=1}$, namely $S_{1}$.
Moreover we introduce the stationary linear Boltzmann equation 
\begin{equation}\label{BNPPL:eq:Boltz2}
\left\{\begin{array}{ll}
\big(v\cdot\nabla_x\big)h_\ep^S(x,v)=\ep^{-\lambda}\,  \text{L}_\ep h_\ep^S(x,v),
&\vspace{0.2cm}\\
h_\ep^S(x,v)=\rho_1,\ \ \ \ \ x\in \{0\}\times \R,\quad v_1>0, 
&\vspace{0.2cm}\\
h_\ep^S(x,v)=\rho_2,\ \ \ \ \ x\in \{L\}\times \R,\quad v_1<0,\  
&
\end{array}\right.
\end{equation}
where $ \text{L}_\ep:=\ep^{-2\alpha} \text{L}$ and $\text{L}$ is the linear Boltzmann operator defined as 
\begin{equation}\label{BNPPL:def:L_ve}
\text{L} f (v)=\mu\int_{-1}^1d\rho\big[f(v')-f(v)\big],\qquad  f\in L^1(S_1)
\end{equation}
with
\begin{equation}\label{BNPPL:scattering}
v'=v-2(\omega\cdot v)\omega
\end{equation}
and $\omega$ is the unit vector bisecting the angle between the incoming velocity $v$ and the outgoing velocity $v'$ as specified in Figure \ref{BNPPL_F1}. 
 
\begin{figure}[ht]
\centering
\includegraphics[scale= 0.15]{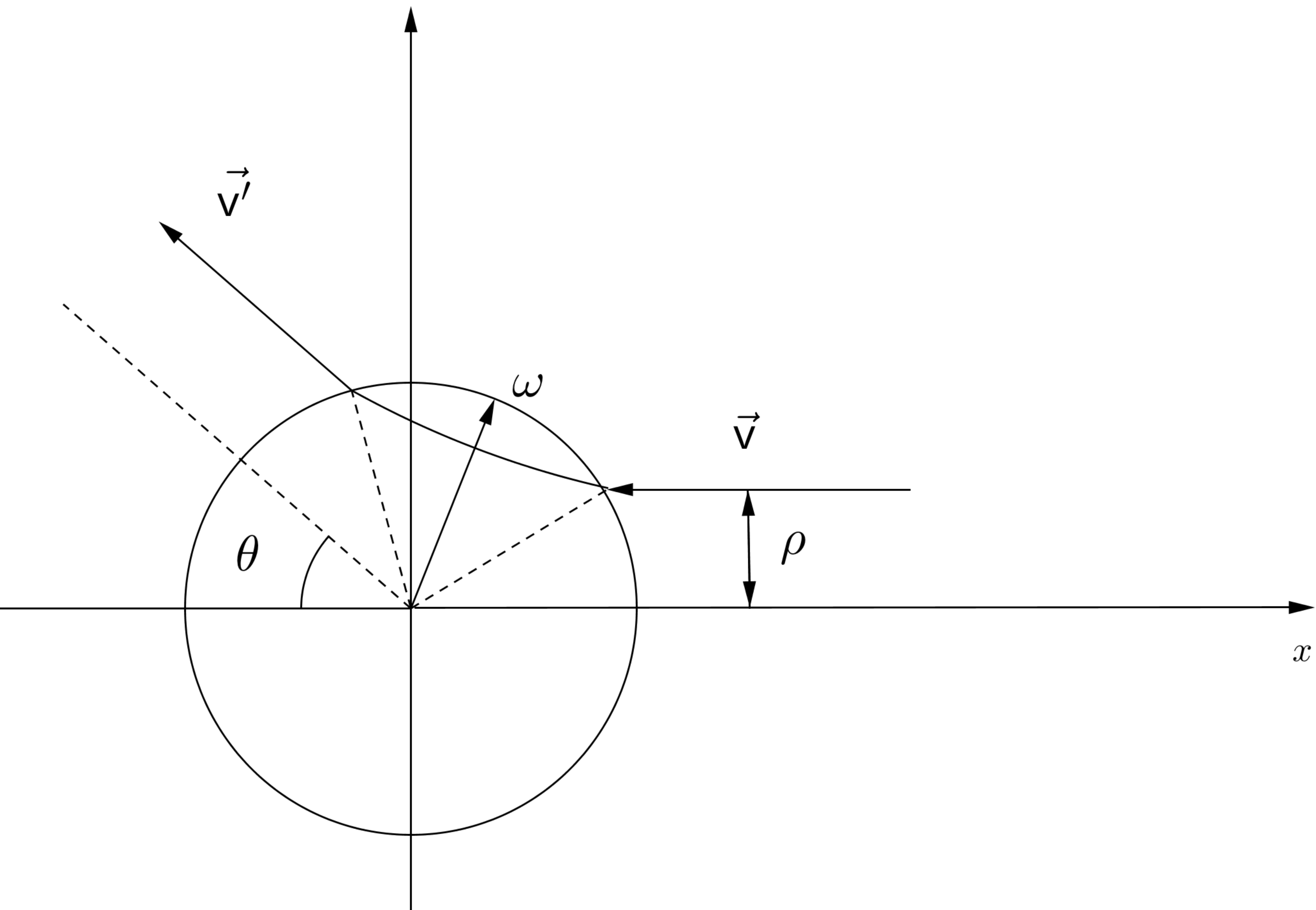}
\caption{The scattering problem}
\label{BNPPL_F1}
\end{figure}

Since the boundary conditions depend on the space variable only trough the horizontal component, the stationary solution $h_\ep^S$ and $g_\ep^S$ inherit the same feature, as well as $f_\ep^S$ and $\varrho^{S}$. 

The strategy of the proof consists of two steps. First we prove that there exists a unique $g_\ep^S$ which converges, as $\ep\to 0$, to $\varrho^S$ given by \eqref{BNPPL:statheat}. See Proposition \ref{BNPPL:prop:hSrhoS} below. Secondly we show that there exists a unique $f_{\ep}^S$ asymptotically equivalent to $g_\ep^S$. See Proposition \ref{BNPPL:prop:fShS} below. This result is achieved using two steps. The first one concerns the convergence of $f_{\ep}^S$ towards $h_{\ep}^S$, the stationary solution of the linear Boltzmann equation, by showing that the memory effects of the mechanical system, preventing the Markovianity, are indeed negligible.
The second one concerns the grazing collision limit which guarantees the asymptotic equivalence of $h_{\ep}^S$ and $g_{\ep}^S$. %

\vspace{2mm}
Let $g_\ep$ be the solution of the  
problem
\begin{equation}\label{BNPPL:eq:Landau1}
\left\{\begin{array}{ll}
\big(\partial_t +v\cdot\nabla_x\big)g_\ep(x,v,t)=\ep^{-\lambda}\, \LL g_\ep(x,v,t),
&\vspace{0.2cm}\\
g_\ep(x,v,0)=f_0(x,v), \ \ \ \ \ \ f_0\in L^{\infty}(\Lambda\times S_1),&\vspace{0.2cm}\\
g_\ep(x,v,t)=\rho_1,\ \ \ \ \ x\in \{0\}\times \R,\quad v_1>0,\ \ \ t\geq 0,\  
&\vspace{0.2cm}\\
g_\ep(x,v,t)=\rho_2,\ \ \ \ \ x\in \{L\}\times \R,\quad v_1<0,\ \ \ t\geq 0.\  
&
\end{array}\right.
\end{equation}

We can write $g_{\ep}(t)$ as the sum of two contributions, one due to the backward trajectories hitting the boundary and the other one due to the trajectories which never leave $\Lambda$.
Therefore we set
\begin{equation*}\label{BNPPL:def:hep}
g_{\ep}(x,v,t)=g_{\ep}^{out}(x,v,t)+g_{\ep}^{in}(x,v,t).
\end{equation*} 
Observe that
$g_{\ep}^{out}$ solves 
\begin{equation}\label{BNPPL:eq:Landau}
\left\{\begin{array}{ll}
\big(\partial_t +v\cdot\nabla_x\big)g_\ep^{out}(x,v,t)=\ep^{-\lambda}\, \LL g_\ep^{out}(x,v,t),
&\vspace{0.2cm}\\
g_\ep^{out}(x,v,0)=0, \ \ \ \ \ \ x\in \Lambda,&\vspace{0.2cm}\\
g_\ep^{out}(x,v,t)=\rho_1,\ \ \ \ \ x\in \{0\}\times \R,\quad v_1>0,\ \ \ t\geq 0,\  
&\vspace{0.2cm}\\
g_\ep^{out}(x,v,t)=\rho_2,\ \ \ \ \ x\in \{L\}\times \R,\quad v_1<0,\ \ \ t\geq 0.\  
&
\end{array}\right.
\end{equation}

\vspace{2mm}
We set $\tilde{\LL}:=\ep^{-\lambda}\, \LL-v\cdot \nabla_x $. Let $G_\ep^0(t)$ be the semigroup whose generator is the operator $\tilde{\LL}$, i.e. $G_\ep^0(t)=e^{t\tilde{\LL}}$. Hence
$$
g_{\ep}^{in}(t)=G_\ep^0(t)f_0.
$$ 
We observe that $g_\ep^S$, solution of \eqref{BNPPL:eq:LandauS}, satysfies, for $t_0>0$
\begin{equation*}
g_\ep^S=g_{\ep}^{out}(t_0)+G_\ep^0(t_0)g_\ep^S,
\end{equation*}
so that we can formally express $g_\ep^S$ as the Neumann series
\begin{equation}\label{BNPPL:eq:gSN}
g_\ep^S=\sum_{n\geq 0}(G_\ep^0(t_0))^n g_\ep^{out}(t_0).
\end{equation}
We now establish existence and uniqueness of $g_\ep^S$ by showing 
that the Neumann series \eqref{BNPPL:eq:gSN} converges. 
In order to do it we extend the action of the semigroup $G_\ep^0(t)$ to the space $L^\infty(\R^2\times S_1)$, namely
$$G_\ep^0(t)\ell_0(x,v)=\chi_{\Lambda}(x) \tilde{G}_\ep^0(t)\ell_0(x,v)$$
for any $\ell_0(x,v)\in L^\infty(\R^2\times S_1).$ Here $\chi_{\Lambda}$ is the characteristic function of $\Lambda$ and $\tilde{G}_\ep^0$ is the extension of the semigroup to the whole space $\R^2\times S_1$. For the sake of simplicity from now on we set $\tilde{G}_\ep^0:=G_\ep^0.$ 

As we proved in \cite{BNPP}, the same technique works for $h_{\ep}$, solution of the following Boltzmann equation \vspace{2mm}
\begin{equation}\label{BNPPL:eq:Boltz1}
\left\{\begin{array}{ll}
\big(\partial_t +v\cdot\nabla_x\big)h_\ep(x,v,t)=\ep^{-\lambda}\, \text{L}_\ep h_\ep(x,v,t),
&\vspace{0.2cm}\\
h_\ep(x,v,0)=f_0(x,v), \ \ \ \ \ \ f_0\in L^{\infty}(\Lambda\times S_1),&\vspace{0.2cm}\\
h_\ep(x,v,t)=\rho_1,\ \ \ \ \ x\in \{0\}\times \R,\quad v_1>0,\ \ \ t\geq 0,\  
&\vspace{0.2cm}\\
h_\ep(x,v,t)=\rho_2,\ \ \ \ \ x\in \{L\}\times \R,\quad v_1<0,\ \ \ t\geq 0.\  
&
\end{array}\right.
\end{equation}

\vspace{3mm}
\noindent The solution $h_\ep$ of the problem \eqref{BNPPL:eq:Boltz1} has the following explicit representation
\vspace{1mm}
\begin{equation}
\begin{split}
\label{BNPPL:eq:heps}
h_{\ep}(x,v,t)&= \sum_{N\geq 0} \left(\mu_\ep\ep\right)^{N}\int_{0}^{t}dt_1\dots\int_{0}^{t_{N-1}}dt_N\\&
\int_{-1}^{1}d\rho_1\dots\int_{-1}^{1}d\rho_N \, \chi(\tau<t_N)\chi(\tau>0)\, e^{-2\mu_\ep\ep \, (t-\tau)}\, f_{B}(\gamma^{-(t-\tau)}(x,v))+\\&
+\sum_{N\geq 0} e^{-2\mu_\ep\ep \, t}\left(\mu_\ep\ep\right)^{N}\int_{0}^{t}dt_1\dots\int_{0}^{t_{N-1}}dt_N\\&
\int_{-1}^{1}d\rho_1\dots\int_{-1}^{1}d\rho_N\, \chi(\tau=0)\, f_0(\gamma^{-t}(x,v)),
\end{split}
\end{equation}
with $f_B$ defined in \eqref{BNPPL:def:fB}. Given $x,v,\,t_1\dots t_N,\,\rho_1\dots\rho_N$, $\gamma^{-t}(x,v)$ denotes the trajectory whose position and velocity are
$$(x-v(t-t_1)-v_1(t_1-t_2)\dots-v_Nt_N,v_N).$$
The transitions $v\to v_1\to v_2\dots\to v_N$ are obtained by means of a scattering with an hard disk with impact parameter $\rho_i$ via \eqref{BNPPL:scattering}.
As before $t-\tau$, $\tau=\tau(x,v,t_1\dots,t_N,\rho_1\dots \rho_N)$, is the first (backward) hitting time with the boundary. We remind that $\mu_\ep\ep=\mu \ep^{-2\alpha-\lambda} $. 

We set
\begin{equation*}\label{BNPPL:def:hep}
h_{\ep}(x,v,t)=h_{\ep}^{out}(x,v,t)+h_{\ep}^{in}(x,v,t).
\end{equation*} 
Observe that
$h_{\ep}^{out}$ solves 
\begin{equation}\label{BNPPL:eq:Boltz}
\left\{\begin{array}{ll}
\big(\partial_t +v\cdot\nabla_x\big)h_\ep^{out}(x,v,t)=\ep^{-\lambda}\, \text{L}_\ep h_\ep^{out}(x,v,t),
&\vspace{0.2cm}\\
h_\ep^{out}(x,v,0)=0, \ \ \ \ \ \ x\in \Lambda,&\vspace{0.2cm}\\
h_\ep^{out}(x,v,t)=\rho_1,\ \ \ \ \ x\in \{0\}\times \R,\quad v_1>0,\ \ \ t\geq 0,\  
&\vspace{0.2cm}\\
h_\ep^{out}(x,v,t)=\rho_2,\ \ \ \ \ x\in \{L\}\times \R,\quad v_1<0,\ \ \ t\geq 0.\  
&
\end{array}\right.
\end{equation}
Let $S_\ep^0(t)$ be the Markov semigroup associated to the second sum in \eqref{BNPPL:eq:heps}, hence $h_{\ep}^{in}(t)=S_\ep^0(t)f_0$. 
Moreover $h_\ep^S$, solution of \eqref{BNPPL:eq:Boltz2}, satysfies, for $t_0>0$
\begin{equation*}
h_\ep^S=h_{\ep}^{out}(t_0)+S_\ep^0(t_0)h_\ep^S,
\end{equation*}
so that we can formally express $h_\ep^S$ as the Neumann series
\begin{equation}\label{BNPPL:eq:gSN}
h_\ep^S=\sum_{n\geq 0}(S_\ep^0(t_0))^n h_\ep^{out}(t_0).
\end{equation}

\begin{proposition}\label{BNPPL:prop:existencegepS}
There exists $\ep_0>0$ 
such that for any $\ep<\ep_0$ and for any $\ell_0\in L^\infty(\R^2\times S_1)$ we have
\begin{equation}\label{BNPPL:S0bounded}
||G_\ep^0(\ep^{-\lambda})\ell_0 ||_\infty\leq \tilde{\beta}\, ||\ell_0||_\infty,\ \ \  \tilde{\beta}<1.
\end{equation}
As a consequence there exists a unique stationary solution $g_\ep^S\in L^\infty (\Lambda\times S_1)$ satisfying (\ref{BNPPL:eq:LandauS}).
\end{proposition}

To prove Proposition \ref{BNPPL:prop:existencegepS} we also need the following result
\begin{proposition}\label{BNPPL:grazcoll1}
For every $\ell_0\in L^\infty(\R^2\times S_1)$ 
 \begin{equation}\label{BNPPL:eq:gepinsepin}
||\left(G_\ep^0(\ep^{-\lambda}t)-S_\ep^0(\ep^{-\lambda}t)\right)\ell_0 ||_\infty\leq C\ep^{2(\alpha-\lambda)}.
\end{equation}
\end{proposition}

\begin{proof}
We look at the evolution of $h_\ep^{in}(\ep^{-\lambda}t)-g_\ep^{in}(\ep^{-\lambda}t)$, 
namely
\begin{equation}\label{BNPPL:evhg}
\big(\partial_t+\ep^{-\lambda}v\cdot\nabla_x \big)\big(h_\ep^{in}-g_\ep^{in})
=\ep^{-2\lambda}\Big(\tilde{\text{L}}_\ep h_\ep^{in}-\LL g_\ep^{in} \Big),
\end{equation}
where $\LL:=\frac{\mu}{2}\Delta_{|v|}$.
We observe that we can write \eqref{BNPPL:evhg} as 
\begin{equation}\label{BNPPL:evhg2}
\big(\partial_t+\ep^{-\lambda}v\cdot\nabla_x \big)\big(h_\ep^{in}-g_\ep^{in})
=\ep^{-2\lambda}\Big[\tilde{\text{L}}_\ep\big(h_\ep^{in}-g_\ep^{in})+\big(\tilde{\text{L}}_\ep-\LL\big) g_\ep^{in} \Big].
\end{equation}
Hence we can consider $\big(\tilde{\text{L}}_\ep-\LL\big) g_\ep^{in} $, in \eqref{BNPPL:evhg2}, as a source term. 
Recalling that
\begin{equation*}
\tilde{\text{L}}_\ep g_\ep^{in}=\mu \ep^{-2\alpha}\int_{-1}^1 d\rho\,
\big[g_\ep^{in}(v')-g_\ep^{in}(v)  \big],
\end{equation*}
we set
\begin{equation*}\begin{split}
&g_\ep^{in}(v')-g_\ep^{in}(v)\\
&=  (v'-v)\cdot\nabla_{|_{S_1}}g_\ep^{in}(v)\\
&\quad+\frac 1 2  (v'-v)\otimes (v'-v)\nabla_{|_{S_1}}\nabla_{|_{S_1}}g_\ep^{in}(v)\\
&\quad+\frac 1 6 (v'-v)\otimes (v'-v)\otimes (v'-v)
\nabla_{|_{S_1}}\nabla_{|_{S_1}}\nabla_{|_{S_1}}g_\ep^{in}(v)
+R_{\ep},\\
\end{split}\end{equation*}
with $R_{\ep}=\mathcal O (|v-v'|^4)$.
Integrating with respect to $v$ and using symmetry arguments we obtain 
\begin{equation*}\begin{split}
\tilde{\text{L}}_\ep g_\ep^{in}=\mu\ep^{-2\alpha}\big\{
\frac 1 2 \Delta_{|v|}g_\ep^{in}\int_{-1}^1 d\rho\,|v'-v|^2
+\int_{-1}^1 d\rho\,R_{\ep}\big\}.
\end{split}
\end{equation*}
Observe that $|v'-v|^2=4\sin^2\frac{\theta_\ep(\rho)}{2}$. (See Figure \ref{BNPPL_F1}). We remind that the scattering angle $$\theta_\ep(\rho)\leq \pi\ep^\alpha \sup_{r\in[0,1]}|r\,\phi'( r)|+\tilde{C}\ep^{2\alpha}$$ and $\max_{\rho\in[0,1]}\theta_\ep(\rho) \leq C\ep^{\alpha}$ (see \cite{DR}, Section 3, for further details). Moreover
$$B:=\lim_{\ep\to 0}\frac{\mu}{2}\ep^{-2\alpha}\int_{-1}^{1}\theta_\ep(\rho)^2d\rho$$
is the diffusion coefficient of the Landau equation, $B<\infty$, hence 
\begin{equation*}\begin{split}
\tilde{\text{L}}_\ep g_\ep^{in}=B\,\Delta_{|v|}g_\ep^{in}+\frac{\mu}{2}\ep^{-2\alpha} \int_{-1}^1 d\rho\,R_{\ep}.
\end{split}
\end{equation*}
Therefore 
\begin{equation}\label{BNPPL:asymequiv}
\big\|\big(\tilde{\text{L}}_\ep-\LL \big)g_\ep^{in}\big\|_\infty \leq C\ep^{2\alpha},
\end{equation}
which vanishes for $\ep\to 0$.

For a smooth reading we set $w_\ep:=h_\ep^{in}-g_\ep^{in}$ and $A_\ep:=\ep^{-2\lambda}\big(\tilde{\text{L}}_\ep-\LL\big) g_\ep^{in}$. Hence \eqref{BNPPL:evhg2} becomes 
\begin{equation*}
\big(\partial_t+\ep^{-\lambda} v\cdot\nabla_x \big)w_\ep
=\ep^{-2\lambda}\tilde{\text{L}}_\ep w_\ep+A_\ep.
\end{equation*}
Let $\tilde{S}_\ep(t):=S^0_\ep(\ep^{-\lambda}t)$ be the semigroup associated to the  generator 
$-\ep^{-\lambda}\big(v\cdot \nabla_x - \ep^{-\lambda}\,\tilde{\text{L}}_\ep\big)$. By equation (\ref{BNPPL:evhg2})
we get
$$
w_\ep(t)=\tilde{S}_\ep(t)w_\ep(0)+\int_0^t ds\ \tilde{S}_\ep(t-s)\, A_{\ep}(s).
$$
Since $w_\ep(0)=0$ we get 
$$
w_\ep(t)=\int_0^t ds\ \tilde{S}_\ep(t-s)\, A_{\ep}(s).
$$
By the usual series expansion for $\tilde{S}_\ep(t)$ we obtain
\begin{equation*}
\begin{split}
w_\ep(x,v,t)&=\int_0^t ds\sum_{N\geq 0} e^{-2 \mu \ep^{-2\alpha-2\lambda} (t-s)}\left(\mu_\ep \ep \right)^{N}\int_{0}^{\ep^{-\lambda}(t-s)}dt_1\dots\int_{0}^{t_{N-1}}dt_N\\&
\int_{-1}^{1}d\rho_1\dots\int_{-1}^{1}d\rho_N\, \chi(\tau=0)\,A_\ep(\gamma^{-\ep^{-\lambda}(t-s)}(x,v),s). 
\end{split}
\end{equation*}
Thanks to \eqref{BNPPL:asymequiv} we have that $A_\ep$ vanishes in the limit, therefore
\begin{equation*}
\|w_\ep( t)\|_\infty\leq T\|A_\ep( t)\|_\infty\leq C\ep^{2\alpha-2\lambda}
\end{equation*}
Hence $h_\ep^{in}$ and $g_\ep^{in}$ are asymptotically equivalent in $L^\infty$.
\end{proof}

\begin{proposition}\label{BNPPL:grazcollOUT}
Let $T>0$. For any $t\in (0,T]$  
\begin{equation}\label{BNPPL:eq:grazcollOUT}
\|h_{\ep}^{out}(\ep^{-\lambda}t)-g_{\ep}^{out}(\ep^{-\lambda}t)\|_{\infty}\leq C\ep^{2(\alpha-\lambda)}
\end{equation}
\end{proposition}
The proof is essentially the same of Proposition \eqref{BNPPL:grazcoll1}, and to let it work we observe that we need the extension procedure discussed in \cite{BNPP}, Section 5, for $h_{\ep}^{out}$.  

\begin{proof} [Proof of Proposition \ref{BNPPL:prop:existencegepS}] 
From Proposition 2.1 in \cite{BNPP}, 
for any $\ell_0\in L^\infty(\R^2\times S_1)$, we have 
\begin{equation}\label{S0bounded}
||S_\ep^0(\ep^{-\lambda})\ell_0 ||_\infty\leq \beta\, ||\ell_0||_\infty,\ \ \ \beta<1.
\end{equation}

Therefore for $\ep$ small enough
\begin{equation}\label{BNPPL:estG01}
\begin{split}
||G_\ep^0(\ep^{-\lambda})\ell_0||_\infty\leq &\, ||(G_\ep^0(\ep^{-\lambda})-S_\ep^0(\ep^{-\lambda}))\ell_0||_\infty +||S_\ep^0(t)\ell_0||_\infty\\&
 \leq  ||(G_\ep^0(\ep^{-\lambda})-S_\ep^0(\ep^{-\lambda}))\ell_0||_\infty +\beta\, ||\ell_0||_\infty
\end{split}
\end{equation}

Hence, using \eqref{BNPPL:eq:gepinsepin} in \eqref{BNPPL:estG01}, we get 
\begin{equation*}\label{BNPPL:estG02}
\begin{split}\displaystyle
||G_\ep^0(\ep^{-\lambda})\ell_0||_\infty\leq & \,||(G_\ep^0(\ep^{-\lambda})-S_\ep^0(\ep^{-\lambda}))\ell_0||_\infty +\beta\, ||\ell_0||_\infty||\\&
\;\leq  \;\tilde{\omega}(\ep)+\beta\, ||\ell_0||_\infty
\,< \tilde{\beta}\,||\ell_0||_\infty,\ \ \  \tilde{\beta}<1.
\end{split}
\end{equation*}
Here $\tilde{\omega}(\ep)=C\,\ep^{2(\alpha-\lambda)}$. 

Finally, since $ \tilde{\beta}<1$, by \eqref{BNPPL:eq:gSN} we get
\begin{equation*}
|| g_\ep^S||_\infty\leq \frac{1}{(1-\tilde{\beta})}\ ||g_\ep^{out}(\ep^{-\lambda})||_\infty\leq \frac{1}{(1- \tilde{\beta})}\, \rho_2.
\end{equation*}
\end{proof}
The last step is the proof of the convergence of $g_\ep^S$ to the stationary solution of the diffusion problem 
\begin{equation}
\label{BNPPL:HEAT1}
\left\{\begin{array}{ll}
\partial_{t}\varrho-D\Delta\varrho=0
& \vspace{2mm}\\
\varrho(x,t)=\rho_1,\ \ \ \ \ x\in \{0\}\times\R,\ \ \ t\geq 0\  
\vspace{0.2cm}\\
\varrho(x,t)=\rho_2,\ \ \ \ \ x\in \{L\}\times\R,\ \ \ t\geq 0,
\end{array}\right.
\end{equation}
with the diffusion coefficient $D$ given by the Green-Kubo formula \eqref{BNPPL:GK}.
We remind that the stationary solution $\varrho^S$ to the problem \eqref{BNPPL:HEAT1} has the following explicit expression 
\begin{equation}\label{BNPPL:rhoSexp}
\varrho^S(x)= \frac{\rho_1(L-x_1)+\rho_2 x_1}{ L},
\end{equation}
where $x=(x_1,x_2)$. 

By using the Hilbert expansion technique in $L^2$ we can prove 
\begin{proposition}\label{BNPPL:prop:hSrhoS}
Let $g_\ep^S\in L^{\infty}((0,L)\times S_1)$ be the solution to the problem \eqref{BNPPL:eq:LandauS}. Then  
\begin{equation}\label{BNPPL:eq:gSrhoS}
g_{\ep}^S\to \varrho^S
\end{equation}
as $\ep\to 0$, where $\varrho^S(x)$ is given by \eqref{BNPPL:rhoSexp}. The convergence is in $L^2((0,L)\times S_1)$.
\end{proposition}
For the proof we refer to \cite{BNPP}, Section 4.2.
This concludes our analysis of the Markov part of the proof.

Recalling the expression (\ref{BNPPL:def:fep}) for the one-particle correlation function $f_\ep$, we introduce a decomposition analogous to those ones used for $g_{\ep}(t)$ and $h_{\ep}(t)$, namely
\begin{equation}\label{BNPPL:def:fepOUT}
f_{\ep}^{out}(x,v,t):=\EE_\varepsilon[f_B (T^{-(t-\tau)}_{\mbf{c}_{N}}(x,v))\chi(\tau>0)]
\end{equation} 
and
\begin{equation}\label{BNPPL:def:fepIN}
f_{\ep}^{in}(x,v,t):= \EE_\ep[f_0 (T^{-t}_{\mbf{c}_{N}}(x,v))\chi(\tau=0)],
\end{equation} 
so that 
\begin{equation*}\label{BNPPL:def:fep2}
f_{\ep}(x,v,t)=f_{\ep}^{out}(x,v,t)+f_{\ep}^{in}(x,v,t).
\end{equation*} 
Here $f_{\ep}^{out}$ is the contribution due to the trajectories that do leave $\Lambda$ at times smaller than $t$, while $f_{\ep}^{in}$ is the contribution due to the trajectories that stay internal to $\Lambda$. We introduce the flow $F_\ep^0(t)$ such that 
\begin{equation*}\label{BNPPL:def:F0}
(F_\ep^0(t)\ell)(x,v)=\EE_\ep[\ell (T^{-t}_{\mbf{c}_{N}}(x,v))\chi(\tau=0)], \quad \ell\in L^{\infty}(\Lambda\times S_1)
\end{equation*} 
and remark that $F_\ep^0$ is just the dynamics ''inside'' $\Lambda$. In particular $f_{\ep}^{in}(t)=F_\ep^0(t)f_0.$

To detect the stationary solution $f_\ep^S$ for the microscopic dynamics we proceed as for the Boltzmann evolution (see \eqref{BNPPL:def:ST}) by setting, for $t_0>0$,
\begin{equation*}
f_\ep^S=f_{\ep}^{out}(t_0)+F_\ep^0(t_0)f_\ep^{S}
\end{equation*}
and we can formally express the stationary solution as the Neumann series
\begin{equation}\label{BNPPL:eq:fSN}
f_\ep^S=\sum_{n\geq 0}(F_\ep^0(t_0))^n f_\ep^{out}(t_0).
\end{equation}
To show the convergence of the series \eqref{BNPPL:eq:fSN} and hence existence of 
$f_\ep^S$ we first need the following Propositions.

\begin{proposition}\label{BNPPL:th:propCIN}
Let $T>0$. For any $t\in (0,T]$  
\begin{equation}\label{BNPPL:eq:CIN}
\|f_{\ep}^{out}(t)-h_{\ep}^{out}(t)\|_{L^\infty(\Lambda\times S_{1} )}\leq C\ep^{\gamma}\, t^3,
\end{equation}
where $h_\ep^{out}$ solves (\ref{BNPPL:eq:Boltz}) and $\gamma=1-8(\alpha-\frac{\lambda}{2})$.
\end{proposition} 
\begin{proposition}\label{BNPPL:prop:fepinhepin}
For every $\ell_0\in L^\infty(\Lambda\times S_1)$  
 \begin{equation}\label{BNPPL:eq:fepinhepin}
||\left(F_\ep^0(t)-S_\ep^0(t)\right)\ell_0 ||_\infty\leq C||\ell_0||_\infty\, \ep^{\gamma}\, t^3,\quad \forall t\in [0,T],
\end{equation}
where $\gamma=1-8(\alpha-\frac{\lambda}{2})$.
\end{proposition}
See Section 5 and Section 6 in \cite{BNP}, and Section 5 in \cite{BNPP} for the proof. As a corollary we can prove 
\begin{proposition}\label{BNPPL:prop:fShS}
For $\ep$ sufficiently small there exists a unique stationary solution $f_\ep^S\in L^\infty(\Lambda\times S_1)$ satisfying \eqref{BNPPL:def:ST}. 
Moreover 
\begin{equation}\label{BNPPL:eq:fShS}
\|f_\ep^S-g_\ep^S\|_\infty\leq C\ep^{\gamma-3\lambda},
\end{equation}
where $\gamma=1-8(\alpha-\frac{\lambda}{2})$.
\end{proposition}
\begin{proof}
We prove the existence and uniqueness of the stationary solution by showing that the Neumann series \eqref{BNPPL:eq:fSN} converges, namely
\begin{equation}\label{BNPPL:F0bounded*}
||F_\ep^0(\ep^{-\lambda})f_0 ||_\infty\leq \beta'\, ||f_0||_\infty,\ \ \ \beta'<1.
\end{equation}
This implies
\begin{equation*}
|| f_\ep^S||_\infty\leq \frac{1}{(1-\beta')}\ ||f_{\ep}^{out}(\ep^{-\lambda})||_\infty\leq \frac{1}{(1-\beta')}\, \rho_2,\ \ \ \ \beta'<1.
\end{equation*}
In fact, since 
\begin{equation*}
||F_\ep^0(\ep^{-\lambda})f_0 ||_\infty\leq ||\left(F_\ep^0(\ep^{-\lambda})-S_\ep^0(\ep^{-\lambda})\right)f_0 ||_\infty+||S_\ep^0(\ep^{-\lambda})f_0 ||_\infty,
\end{equation*}
thanks to \ref{BNPPL:prop:fepinhepin} and Propositions 2.1 in \cite{BNPP} we get 
\begin{equation}\label{BNPPL:F0bounded}
\begin{split}
||F_\ep^0(\ep^{-\lambda})f_0 ||_\infty\leq &\,||f_0||_\infty C\ep^{\gamma-3\lambda} +  ||S_\ep^0(\ep^{-\lambda})f_0 ||_\infty\\ \leq &\, (C\ep^{\gamma-3\lambda}+\beta)||f_0||_\infty
\leq \beta'||f_0||_\infty, 
\end{split}
\end{equation}
with $\beta'<1$, for $\ep$ sufficiently small (remind that $\ep^{\gamma-3\lambda}\to 0$ as $\ep\to 0$).
This guarantees the existence and uniqueness of the microscopic stationary solution $f_{\ep}^S$.

In order to prove $\eqref{BNPPL:eq:fShS}$ we observe that 
\begin{equation*}
\|f^S_{\ep}-g^S_{\ep}\|_{\infty}\leq \|f^S_{\ep}-h^S_{\ep}\|_{\infty}+\|h^S_{\ep}-g^S_{\ep}\|_{\infty}.
\end{equation*}
We compare the two Neumann series representing $f_{\ep}^S$ and $h_{\ep}^S$,
\begin{equation}\label{BNPPL:tbyt}
\begin{split}
\|f^S_{\ep}-h^S_{\ep}\|_{\infty}=&
\,\|\sum_{n\geq 0}\big((F_\ep^0(\ep^{-\lambda}))^n f_{\ep}^{out}(\ep^{-\lambda}) -(S_\ep^0(\ep^{-\lambda}))^n h_{\ep}^{out}(\ep^{-\lambda})\big)\|_{\infty}\\&
\leq \sum_{n\geq 0}\|(F_\ep^0(\ep^{-\lambda}))^n(f_{\ep}^{out}(\ep^{-\lambda})-h_{\ep}^{out}(\ep^{-\lambda}))\|_{\infty}\\&+
\sum_{n\geq 0}\|\big((F_\ep^0(\ep^{-\lambda}))^n  -(S_\ep^0(\ep^{-\lambda}))^n\big) h_{\ep}^{out}(\ep^{-\lambda})\|_{\infty}.
\end{split}
\end{equation}
By \eqref{BNPPL:F0bounded}, using Proposition \ref{BNPPL:th:propCIN}, the first sum on the right hand side of \eqref{BNPPL:tbyt} is bounded by
$$\frac{1}{1-\beta'}\|f_{\ep}^{out}(\ep^{-\lambda})-h_{\ep}^{out}(\ep^{-\lambda})\|_{\infty}\leq C\ep^{\gamma-3\lambda}.$$
As regard to the second sum on the right hand side of \eqref{BNPPL:tbyt} we have
\begin{equation}\label{BNPPL:trick}
\begin{split}
&\sum_{n\geq 0}\|\big((F_\ep^0(\ep^{-\lambda})^n  -(S_\ep^0(\ep^{-\lambda}))^n\big) h_{\ep}^{out}(\ep^{-\lambda})\|_{\infty}\\&
\leq \sum_{n\geq 0}\sum_{k=0}^{n-1}\|(F_\ep^0(\ep^{-\lambda}))^{n-k-1} \big(F_\ep^0(\ep^{-\lambda})-S_\ep^0(\ep^{-\lambda})\big) (S_\ep^0(\ep^{-\lambda}))^k h_{\ep}^{out}(\ep^{-\lambda})\|_{\infty}\\&
\leq \sum_{k,\ell\geq 0}\|(F_\ep^0(\ep^{-\lambda}))^{\ell} \big(F_\ep^0(\ep^{-\lambda})-S_\ep^0(\ep^{-\lambda})\big) (S_\ep^0(\ep^{-\lambda}))^k h_{\ep}^{out}(\ep^{-\lambda})\|_{\infty}\\&
\leq C\,\|h_{\ep}^{out}(\ep^{-\lambda})\|_{\infty}\,\ep^{\gamma-3\lambda},
 \end{split}
 \end{equation}
by virtue of \eqref{BNPPL:S0bounded}, \eqref{BNPPL:F0bounded} and \eqref{BNPPL:eq:fepinhepin}. 

We compare the two Neumann series representing $h_{\ep}^S$ and $g_{\ep}^S$,
\begin{equation}\label{BNPPL:tbyt2}
\begin{split}
\|h^S_{\ep}-g^S_{\ep}\|_{\infty}=&
\,\|\sum_{n\geq 0}\big((S_\ep^0(\ep^{-\lambda}))^n h_{\ep}^{out}(\ep^{-\lambda}) -(G_\ep^0(\ep^{-\lambda}))^n g_{\ep}^{out}(\ep^{-\lambda})\big)\|_{\infty}\\&
\leq \sum_{n\geq 0}\|(S_\ep^0(\ep^{-\lambda}))^n(h_{\ep}^{out}(\ep^{-\lambda})-g_{\ep}^{out}(\ep^{-\lambda}))\|_{\infty}\\&+
\sum_{n\geq 0}\|\big((S_\ep^0(\ep^{-\lambda}))^n  -(G_\ep^0(\ep^{-\lambda}))^n\big) g_{\ep}^{out}(\ep^{-\lambda})\|_{\infty}.
\end{split}
\end{equation}
By using Proposition \ref{BNPPL:grazcollOUT} the first sum on the right hand side of \eqref{BNPPL:tbyt2} is bounded by
$$\frac{1}{1-\beta'}\|h_{\ep}^{out}(\ep^{-\lambda})-g_{\ep}^{out}(\ep^{-\lambda})\|_{\infty}\leq C\ep^{2(\alpha-\lambda)}.$$


\noindent As regard to the second sum on the right hand side of \eqref{BNPPL:tbyt2} by means of the same trick used in \eqref{BNPPL:trick} we get
\begin{equation*}
\sum_{n\geq 0}\|\big((S_\ep^0(\ep^{-\lambda}))^n  -(G_\ep^0(\ep^{-\lambda}))^n\big) g_{\ep}^{out}(\ep^{-\lambda})\|_{\infty}\\
\leq C\,\|g_{\ep}^{out}(\ep^{-\lambda})\|_{\infty}\,\ep^{2(\alpha-\lambda)}.
 \end{equation*}

This concludes the proof of Proposition \ref{BNPPL:prop:fShS}.
\end{proof}
Hence the proof of Theorem \ref{BNPPL:th:MAIN1} follows from Proposition 
\ref{BNPPL:prop:hSrhoS} and Proposition \ref{BNPPL:prop:fShS}.
We conclude by proving Theorem \ref{BNPPL:th:MAIN2} which actually is a Corollary of the previous analysis.
\begin{proof} [Proof of Theorem \ref{BNPPL:th:MAIN2}]
By standard computations (see e.g. Section \cite{BNPP}, Section 4.2) we have
$$
g_\ep^S=\varrho^S+\frac{1}{\ep^{-\lambda}}g^{(1)}+\frac{1}{\ep^{-\lambda}}R_{\ep},
$$
where 
$$
g^{(1)}(v)=\LL^{-1}(v\cdot\nabla_x\varrho^S)=\frac{\rho_2-\rho_1}{L}\LL^{-1}(v_1)
$$
and, as we see in \cite{BNPP}, Section 4.2, $R_{\ep}=O(\ep^{\frac{\lambda}{2}})$ in $L^{2}((0,L)\times S_1).$
Therefore, since $\int_{S_1} v\varrho^S dv=0$, 
\begin{equation}\label{BNPPL:Fp}
\ep^{-\lambda}\int_{S_1} v g_\ep^S(x,v)dv=D\nabla_x\varrho^S+O(\ep^{\frac{\lambda}{2}}),
\end{equation} 
where $D$ is given by \eqref{BNPPL:GK}. By Theorem \ref{BNPPL:th:MAIN1} the right hand side of \eqref{BNPPL:Fp} is close to $D\nabla_x\varrho_{\ep}^S$ in $\mathcal{D}'((0,L)\times S_1)$, where $\varrho_{\ep}^S$ is given by \eqref{BNPPL:def:mass}. On the other hand, by Proposition \ref{BNPPL:prop:fShS} and Assumption \ref{BNPPL:A1}, the left hand side of \eqref{BNPPL:Fp} is close in $L^{\infty}((0,L)\times S_1)$ to $J_{\ep}^S(x)$ defined in \eqref{BNPPL:def:j}. This concludes the proof of \eqref{BNPPL:FickL}. Moreover \eqref{BNPPL:Jlim} and \eqref{BNPPL:FICK} follow by \eqref{BNPPL:Fp}.
\end{proof}


\end{document}